\documentclass[pra,aps,twocolumn,showpacs]{revtex4}

\usepackage{amsmath,amsfonts,amssymb,caption,hyperref,color,epsfig,graphics,graphicx,latexsym,mathrsfs,revsymb,theorem,url,epstopdf}

\hypersetup{colorlinks,linkcolor={blue},citecolor={blue},urlcolor={red}}

\usepackage{hyperref}
\usepackage{float}
\usepackage{subfigure}

\newlength{\halfpagewidth}

\divide\halfpagewidth by 2

\newtheorem{definition}{Definition}
\newtheorem{proposition}[definition]{Proposition}
\newtheorem{Lemma}[definition]{Lemma}

\newtheorem{Theorem}[definition]{Theorem}

\newtheorem{conjecture}[definition]{Conjecture}

\newtheorem{remark}[definition]{Remark}
\newtheorem{example}[definition]{Example}
\newtheorem{question}[definition]{Question}

\def\squareforqed{\hbox{\rlap{$\sqcap$}$\sqcup$}}
\def\qed{\ifmmode\squareforqed\else{\unskip\nobreak\hfil
		\penalty50\hskip1em\null\nobreak\hfil\squareforqed
		\parfillskip=0pt\finalhyphendemerits=0\endgraf}\fi}
\def\endenv{\ifmmode\;\else{\unskip\nobreak\hfil
		\penalty50\hskip1em\null\nobreak\hfil\;
		\parfillskip=0pt\finalhyphendemerits=0\endgraf}\fi}
\newenvironment{proof}{\noindent \textbf{{Proof.~} }}{\qed}
\def\Dbar{\leavevmode\lower.6ex\hbox to 0pt
	{\hskip-.23ex\accent"16\hss}D}
\makeatletter
\def\url@leostyle{%
	\@ifundefined{selectfont}{\def\UrlFont{\sf}}{\def\UrlFont{\small\ttfamily}}}
\makeatother
\def\bcj{\begin{conjecture}}
	\def\ecj{\end{conjecture}}
\def\bcr{\begin{corollary}}
	\def\ecr{\end{corollary}}
\def\bd{\begin{definition}}
	\def\ed{\end{definition}}
\def\bea{\begin{eqnarray}}
\def\eea{\end{eqnarray}}
\def\bem{\begin{enumerate}}
	\def\eem{\end{enumerate}}
\def\bex{\begin{example}}
	\def\eex{\end{example}}
\def\bim{\begin{itemize}}
	\def\eim{\end{itemize}}
\def\bl{\begin{lemma}}
	\def\el{\end{lemma}}
\def\bma{\begin{bmatrix}}
	\def\ema{\end{bmatrix}}
\def\bpf{\begin{proof}}
	\def\epf{\end{proof}}
\def\bpp{\begin{proposition}}
	\def\epp{\end{proposition}}
\def\bqu{\begin{question}}
	\def\equ{\end{question}}
\def\br{\begin{remark}}
	\def\er{\end{remark}}
\def\bt{\begin{theorem}}
	\def\et{\end{theorem}}

\def\btb{\begin{tabular}}
	\def\etb{\end{tabular}}

\newcommand{\nc}{\newcommand}

\nc{\bbA}{\mathbb{A}} \nc{\bbB}{\mathbb{B}} \nc{\bbC}{\mathbb{C}}
\nc{\bbD}{\mathbb{D}} \nc{\bbE}{\mathbb{E}} \nc{\bbF}{\mathbb{F}}
\nc{\bbG}{\mathbb{G}} \nc{\bbH}{\mathbb{H}} \nc{\bbI}{\mathbb{I}}
\nc{\bbJ}{\mathbb{J}} \nc{\bbK}{\mathbb{K}} \nc{\bbL}{\mathbb{L}}
\nc{\bbM}{\mathbb{M}} \nc{\bbN}{\mathbb{N}} \nc{\bbO}{\mathbb{O}}
\nc{\bbP}{\mathbb{P}} \nc{\bbQ}{\mathbb{Q}} \nc{\bbR}{\mathbb{R}}
\nc{\bbS}{\mathbb{S}} \nc{\bbT}{\mathbb{T}} \nc{\bbU}{\mathbb{U}}
\nc{\bbV}{\mathbb{V}} \nc{\bbW}{\mathbb{W}} \nc{\bbX}{\mathbb{X}}
\nc{\bbZ}{\mathbb{Z}}

\nc{\bA}{{\bf A}} \nc{\bB}{{\bf B}} \nc{\bC}{{\bf C}}
\nc{\bD}{{\bf D}} \nc{\bE}{{\bf E}} \nc{\bF}{{\bf F}}
\nc{\bG}{{\bf G}} \nc{\bH}{{\bf H}} \nc{\bI}{{\bf I}}
\nc{\bJ}{{\bf J}} \nc{\bK}{{\bf K}} \nc{\bL}{{\bf L}}
\nc{\bM}{{\bf M}} \nc{\bN}{{\bf N}} \nc{\bO}{{\bf O}}
\nc{\bP}{{\bf P}} \nc{\bQ}{{\bf Q}} \nc{\bR}{{\bf R}}
\nc{\bS}{{\bf S}} \nc{\bT}{{\bf T}} \nc{\bU}{{\bf U}}
\nc{\bV}{{\bf V}} \nc{\bW}{{\bf W}} \nc{\bX}{{\bf X}}
\nc{\bZ}{{\bf Z}}

\nc{\cA}{{\cal A}} \nc{\cB}{{\cal B}} \nc{\cC}{{\cal C}}
\nc{\cD}{{\cal D}} \nc{\cE}{{\cal E}} \nc{\cF}{{\cal F}}
\nc{\cG}{{\cal G}} \nc{\cH}{{\cal H}} \nc{\cI}{{\cal I}}
\nc{\cJ}{{\cal J}} \nc{\cK}{{\cal K}} \nc{\cL}{{\cal L}}
\nc{\cM}{{\cal M}} \nc{\cN}{{\cal N}} \nc{\cO}{{\cal O}}
\nc{\cP}{{\cal P}} \nc{\cQ}{{\cal Q}} \nc{\cR}{{\cal R}}
\nc{\cS}{{\cal S}} \nc{\cT}{{\cal T}} \nc{\cU}{{\cal U}}
\nc{\cV}{{\cal V}} \nc{\cW}{{\cal W}} \nc{\cX}{{\cal X}}
\nc{\cZ}{{\cal Z}}

\nc{\hA}{{\hat{A}}} \nc{\hB}{{\hat{B}}} \nc{\hC}{{\hat{C}}}
\nc{\hD}{{\hat{D}}} \nc{\hE}{{\hat{E}}} \nc{\hF}{{\hat{F}}}
\nc{\hG}{{\hat{G}}} \nc{\hH}{{\hat{H}}} \nc{\hI}{{\hat{I}}}
\nc{\hJ}{{\hat{J}}} \nc{\hK}{{\hat{K}}} \nc{\hL}{{\hat{L}}}
\nc{\hM}{{\hat{M}}} \nc{\hN}{{\hat{N}}} \nc{\hO}{{\hat{O}}}
\nc{\hP}{{\hat{P}}} \nc{\hR}{{\hat{R}}} \nc{\hS}{{\hat{S}}}
\nc{\hT}{{\hat{T}}} \nc{\hU}{{\hat{U}}} \nc{\hV}{{\hat{V}}}
\nc{\hW}{{\hat{W}}} \nc{\hX}{{\hat{X}}} \nc{\hZ}{{\hat{Z}}}

\nc{\hn}{{\hat{n}}}

\def\max{\mathop{\rm max}}
\def\min{\mathop{\rm min}}

\def\tr{\mathop{\rm Tr}}

\newcommand{\bra}[1]{\langle#1|}
\newcommand{\ket}[1]{|#1\rangle}

\newcommand{\norm}[1]{\lVert#1\rVert}

\def\Dbar{\leavevmode\lower.6ex\hbox to 0pt
	{\hskip-.23ex\accent"16\hss}D}

\begin{document}
	\title{A Genuine Multipartite Entanglement Measure Generated by the Parametrized Entanglement Measure }
	
	\author{Xian Shi}\email[]
	{shixian01@gmail.com}
\affiliation{College of Information Science and Technology, Beijing University of Chemical Technology, Beijing 100029, China}
	
	%
	
	
	
	\date{\today}
	

\begin{abstract}
\indent In this paper, we investigate a genuine multipartite entanglement measure based on the geometric method. This measure arrives at the maximal value for the absolutely maximally entangled states and has desirable properties for quantifying the genuine multipartite entanglement. We present a lower bound of the genuine multipartite entanglement measure. At last, we present some examples to show that the genuine entanglement measure is with distinct entanglement ordering from other measures, and we also present the advantages of the measure proposed here with other measures.
\end{abstract}

\maketitle
\section{introduction}
 \indent \indent Quantum entanglement is an essential feature of quantum mechanics. It plays an important role in quantum information and quantum computation theory \cite{horodecki2009quantum}, such as superdense coding \cite{bennett1992communication}, teleportation \cite{bennett1993teleporting} and the speedup of quantum algorithms \cite{shimoni2005entangled}. \\
\indent  One of the most important problems is to quantify the entanglement in a composite quantum system. Vedral $et$ $al.$ in \cite{vedral1997quantifying} presented the condition that the amount of entanglement cannot increase under local operation and classical communication (LOCC) is necessary for an entanglement measure. Then Vidal considered the entanglement measures with stronger properties and proposed a general mathematical framework to build entanglement monotone with functions satisfying some properties for pure states\cite{vidal2000entanglement}. The other interesting approach with operational significance to quantify the entanglement is proposed in \cite{gour2020optimal,shi2021extension}. Compared with the bipartite entanglement systems, the complexity of a multipartite entanglement system grows remarkly with the increasing number of parties and the increasing dimension of the systems. The notion of multipartite entanglement measure can be refined into the so-called genuine multipartite entanglement (GME) \cite{guhne2005multipartite,plenio2014introduction}. Substantial results have been achiveved on the GME measures in the last few decades \cite{meyer2002global,blasone2008hierarchies,hiesmayr2009two,ma2011measure,jungnitsch2011taming,rafsanjani2012genuinely,xie2021triangle,beckey2021computable,guo2022genuine,li2022geometric}. In \cite{blasone2008hierarchies}, the authors proposed a genuinely entangled measure defined as the shortest distance from a given state to the k-separable states, which was denoted as generalized geometric measure (GGM). The other GME measure, the genuinely multipatite concurrence (GMC) was defined as the minimal bipartite concurrence among all bipartitions \cite{ma2011measure}. However, the above two measures cannot show all the conditions of entanglement among the parties, both two are defined on the minimizations of the partitions. Concurrence fill was proposed as a three-qubit GME measure \cite{xie2021triangle}, it is the square of the area of the three-qubit concurrence triangle. The authors in \cite{guo2022genuine} generalized the above method to build a multipartite entanglement measure, however, it is hard when the parties are bigger. Hence there is much work to do on how to understand and quantify the multipartite entanglement. Recently, another method to investigate the GME measure is proposed, it was based on the geometric mean of all bipartite entanglement measure concurrence \cite{li2022geometric}.\\
\indent In this paper, we investigate a GME measure which was based on the geometric mean in terms of a bipartite entanglement measure, $F_q$ \cite{yang2021parametrized}. This measure satisfies the following properties, subadditivity, continuity for pure states. We also present the bound of this GME measure based on the method proposed in \cite{dai2020experimentally}. At last, we present some examples to show that the GME measure proposed here is with different orders from GMC and GGM, we also prsent some advantages of the measure when comparing with GMC and GGM.\\
\indent This paper is organised as follows. In Sec. \MakeUppercase{\romannumeral2}, we present the preliminary knowledge needed here. In Sec. \MakeUppercase{\romannumeral3}, we present the main results. We show that the measure is a GME measure, then we consider the measure for the W states and GHZ states in $n$-qubit systems and present that the GHZ states is more entangled than the W states. We also present the measure satifies the subadditivity and continuity for pure states. And then we present a lower bound of the measure for multipartite mixed states. We also make some comparisons between the measure proposed here and the GMC, GGM by considering some examples. In Sec. \MakeUppercase{\romannumeral4}, we end with a conclusion.

\section{Preliminary Knowledge}\label{se2}
\indent Concurrence is one of the most important entanglement measures for bipartite quantum systems \cite{hill1997entanglement}, it has attracted much attention of the relevant researchers since the end of the last century \cite{wootters1998entanglement,coffman2000distributed,mintert2004concurrence,chen2005concurrence,zhang2016evaluation,dai2020experimentally,li2020improved}. For a bipartite pure state $\ket{\psi}_{AB},$ its concurrence is defined as 
\begin{align}
C(\ket{\psi}_{AB})=\sqrt{2(1-Tr_A\ket{\psi}\bra{\psi})},
\end{align}
when $\rho_{AB}$ is a mixed state, its concurrence is defined as
\begin{align}
C(\rho_{AB})=\min\sum_i p_iC(\ket{\psi}_i),
\end{align}
where the minimum takes over all the decompositions of $\rho_{AB}=\sum_ip_i\ket{\psi_i}\bra{\psi_i}.$ For two qubit mixed states, there exists a direct link between the concurrence and the entanglement of formation \cite{wootters1998entanglement}.\par 
 As a generalized von Neumann entropy, Tsallis-$q$ entropy \cite{tsallis1988possible,landsberg1998distributions} can present more properties of the entangled states \cite{tsallis2001peres,rossignoli2002generalized}. For a pure state $\ket{\psi}_{AB}=\sum_i\sqrt{\lambda_i}\ket{ii},$ then its Tsallis-$q$ entanglement measure \cite{san2010tsallis} is defined as
\begin{align*}
T_q(\ket{\psi}_{AB})=\frac{1-\rho_A^q}{q-1},
\end{align*}
here $\rho_A=Tr_B\rho_{AB},$ $q\in(0,1)\cup(1,\infty).$\par 
Motived by the Tsallis-$q$ entanglement entropy, Yang $et$ $al.$ proposed a parametrized entanglement measure, $q$-concurrence, for a bipartite entanglement systems \cite{yang2021parametrized}. When $\ket{\psi}_{AB}$ is a pure state, this measure is defined as
\begin{align}
\mathcal{C}_q(\ket{\psi}_{AB})=F_q(\rho_A),
\end{align} 
here $F(\rho)=1-\tr\rho^q$, $q\ge 2.$ When $\rho_{AB}$ is a mixed state, its $q$-concurrence is defined as
\begin{align}
\mathcal{C}_q(\rho_{AB})=\min\sum_i p_i\mathcal{C}_q(\ket{\psi_i}),
\end{align}
where the minimum takes over all the decompositions of $\rho_{AB}=\sum_ip_i\ket{\psi_i}\bra{\psi_i}.$\par 
Next we recall some properties of the function $F_q(\rho)$ proposed in \cite{yang2021parametrized}.
\begin{Lemma}
For any density matrix $\rho$ and $q\ge 2$, $F_q(\rho)$ satisfies the following properties:
\begin{itemize}
	\item[(i).] \emph{Non-negativity:} $F_q(\rho)\ge 0.$
	\item[(ii).] \emph{Subadditivity:} For a general bipartite state $\rho_{AB}$, $F_q(\rho_{AB})$ satisfies the inequalies:
	\begin{align}
	|F_q(\rho_A)-F_q(\rho_B)|\le F_q(\rho_{AB})\le F_q(\rho_A)+F_q(\rho_B).
	\end{align}
	\item[(iii).] \emph{Concavity} 
	\begin{align}
	\sum p_i F_q(\rho_i)\le F_q(\sum_i p_i\rho_i),
	\end{align}
	where $p_i\in(0,1],$ and $\sum_ip_i=1$, $\rho_i$ are density matrices. Furthermore, The equality holds if and only if $\rho_i$ are identical for all $p_i>0.$
	\item[(iv)] \emph{quasiconvex:}
	\begin{align}
	F_q(\sum_i p_i\rho_i)\le \sum p_i^q F_q(\rho_i)+1-\sum_ip_i^q,
	\end{align}
	where the inequality holds if and only if $\rho_i=\ket{\psi_i}\bra{\psi_i},$ and $\{\ket{\psi_i}\}$ are orthogonal.
 \end{itemize}
\end{Lemma}
\par
Due to the properties of $F_q(\rho)$ above, we have that the maximum of $F_q(\ket{\psi}_{AB})$ is attained when its reduced density matrix of the smaller subsystem is the maximally mixed state, that is, $$\max_{\ket{\psi}_{AB}}F_q(\ket{\psi}_{AB})=\frac{d^{q-1}-1}{d^{q-1}},$$
 here $d$ is the dimension of the smaller system. \par
Next we review the knowledge needed on multipartite entanglement.\par
An $n$-partite pure state $\ket{\psi}_{A_1A_2\cdots A_n}$ is full product if it can be written as
\begin{align}
\ket{\psi}_{A_1A_2\cdots A_n}=\ket{\phi_1}_{A_1}\ket{\phi_2}_{A_2}\cdots \ket{\phi_n}_{A_n},
\end{align}
otherwise, it is entangled. A multipartite pure state is called
genuinely entangled if
\begin{align}
\ket{\psi}_{A_1A_2\cdots A_n}\ne \ket{\phi}_S\ket{\zeta}_{\overline{S}}, \label{bs}
\end{align}
for any partite $S|\overline{S}$, here $S$ is a subset of $\boldsymbol{A}=\{A_1,A_2,\cdots, A_n\}$, and $\overline{S}=\boldsymbol{A}-S$. An $n$ partite mixed state $\rho$ is biseparable if it can be written as a convex combination of biseparable pure states $\rho=\sum_i p_i\ket{\psi_i}\bra{\psi_i},$ where the pure states in $\{\ket{\psi_i}\}$ can be biseparable with respect to some bipartitions. If an $n$-partite state is not biseparable, then it is genuinely entangled. \par
Then we recall the necessary conditions of a genuine multipartite entanglement (GME) measure $E$ should satisfy \cite{ma2011measure}:
\begin{itemize}
\item[1]. it is entanglement monotone.
\item[2]. $E(\rho)=0,$ if $\rho$ is biseparable.
\item[3]. $E(\rho)>0$, if $\rho$ is genuinely entangled state.
\end{itemize}\par
Based on the $q$-concurrence, we present a GME measure in terms of the geometric methods.
 \begin{definition}
Assume $\ket{\psi}_{A_1A_2\cdots A_n}$ is an $n$-partite pure state, the geometric mean of $q$-concurrence $(GqC)$ is defined as
\begin{align}
\mathcal{G}_q(\ket{\psi})=[\mathcal{P}_q(\ket{\psi})]^{\frac{1}{c(\alpha)}},
\end{align}
where $\alpha=\{\alpha_i\}$ is the set that denotes all possible bipartitions $\{A_{\alpha_i}|B_{\alpha_i}\}$ of the $n$ parties, $c(\alpha)$ is the cardinality of $\alpha$, and $\mathcal{P}_q(\ket{\psi})$ is
$$
\mathcal{P}_q(\ket{\psi})=\Pi_{\alpha_i\in\alpha}{\mathcal{C}_q}_{A_{\alpha_i}B_{\alpha_i}}(\ket{\psi}),$$
\begin{align*}
	\begin{split}
	c(\alpha)=\left\{
	\begin{array}{lr}
\sum_{m=1}^{\frac{n-1}{2}}C_n^m,\hspace{3mm} \textit{if $n$ is odd,}\\
\sum_{m=1}^{\frac{n-2}{2}}C_n^m+\frac{1}{2}C_n^{\frac{n}{2}},\hspace{3mm}\textit{if $n$ is even.}
\end{array}
\right.
\end{split}
\end{align*}
\par
When $\rho$ is an $n$-partite mixed state, 
\begin{align}
\mathcal{G}_q(\rho)=\min\sum_i p_i\mathcal{G}_q(\ket{\psi_i}),
\end{align}
where the minimum takes over all the decompositions of $\rho=\sum_i p_i\ket{\psi_i}\bra{\psi_i}.$
 \end{definition}

\section{MAIN RESULTS}
\indent In this section, we present the main results of this article. In Sec. \ref{s1},  we present the properties of the GqC. In Sec. \ref{s2}, a lower bound of the GqC for an $n$-partite mixed state $\rho$ was presented. In Sec. \ref{s3}, we make a comparison of the GqC with other GME measures.
\subsection{The properties of GqC}\label{s1}
Here we first show that the GqC is a GME measure.
\begin{Theorem}
	For an arbitrary $n$-partite quantum state $\ket{\psi}_{A_1A_2\cdots A_n}$, the GqC is a GME measure.
\end{Theorem}
\begin{proof}
	Here we show that the GqC satisfies the properties $2$ and $3$. The proof of the property 1 is similar to the proof in \cite{li2022geometric}, and we omit it here.\par
	Assume $\ket{\psi}$ is a biseparable pure state, then according to the definition of biseparable states in Sec. \ref{se2}, there exists a partition $S|\overline{S}$ of $\boldsymbol{A}$ such tat $\ket{\psi}=\ket{\phi_1}_S\ket{\phi_2}_{\overline{S}}$, $\mathcal{C}_q(\ket{\psi}_{S|\overline{S}})=0$, thus we have $\mathcal{G}_q(\ket{\psi})=0.$ And due to the definition of the GqC for mixed states, when a mixed state $\rho$ is biseparable, $\mathcal{G}_q(\rho)=0,$ Hence we present the proof of condition 1. \par
As a GME pure state $\ket{\psi}$ can be written as (\ref{bs}), that is, $\ket{\psi}$ cannot be written as product states with respect to any bipartition, then we have all the bipartite $\mathcal{C}_q$ is bigger than 0, so $\mathcal{G}_q(\ket{\psi})>0.$ And due to the definition of a mixed GME state $\rho$ and the definition of GqC, we have when $\rho$ is GME, $\mathcal{G}_q(\rho>0),$ then we prove the condition $2.$
 \end{proof}
\par
A pure multipartite entangled state is called absolutely maximally entangled state (AMES) if all reduced density operators obtained by tracing out at least half of the particles of the pure state are maximally mixed \cite{helwig2012absolute}. The AMES can be used to develop the quantum secret sharing shemes \cite{helwig2013absolutely} and quantum error correction codes \cite{grassl2015quantum,alsina2021absolutely}.  And the GHZ state is the only AMES up to the local unitary operations in three qubit systems \cite{goyeneche2015absolutely,shi2021multilinear}.  Due to the definitions and properties of $\mathcal{C}_q(\ket{\psi})$,  when $\ket{\psi}$ is an AMES, then $\mathcal{G}_q(\cdot)$ gets the maximum, that is, it can also be seen as a proper GME measure \cite{li2022geometric}.\\
\indent Next we consider two pure states in multipartite systems that are inequivalent in terms of stochastic LOCC (SLOCC), the W states and GHZ states.  The authors in \cite{joo2003quantum} showed that in three qubit systems, a perfect teleportation can be performed via the GHZ
state, while the W state cannot. Moreover, for a $k$-partite W and GHZ state ($k\ge 3$), the infimum asymptotic ratio from GHZ to W is 1, however, the infimum asymptotic ratio from W to GHZ is bigger than 1 \cite{vrana2015asymptotic}. Thus the GHZ states can be thought more entangled than the W states.
\begin{example}
	\begin{align*}
	\ket{W_n}=&\frac{1}{\sqrt{n}}(\ket{10\cdots0}+\ket{01\cdots0}+\cdots+\ket{00\cdots1}),\\
	\ket{GHZ_n}=&\frac{1}{\sqrt{2}}(\ket{00\cdots0}+\ket{11\cdots1}),
	\end{align*}
\par
Here we place the results on $\mathcal{G}_q(\ket{W_n})$ and $\mathcal{G}_q(\ket{GHZ_n})$ in Sec. \ref{ap1}. In Fig. \ref{fig3}, we can see the values of $\frac{\mathcal{G}_3(\ket{W_n})}{\mathcal{G}_3(\ket{GHZ_n})}$ tends to 1 with the increase of $n$.
	\begin{figure}
		\centering
		\includegraphics[width=90mm]{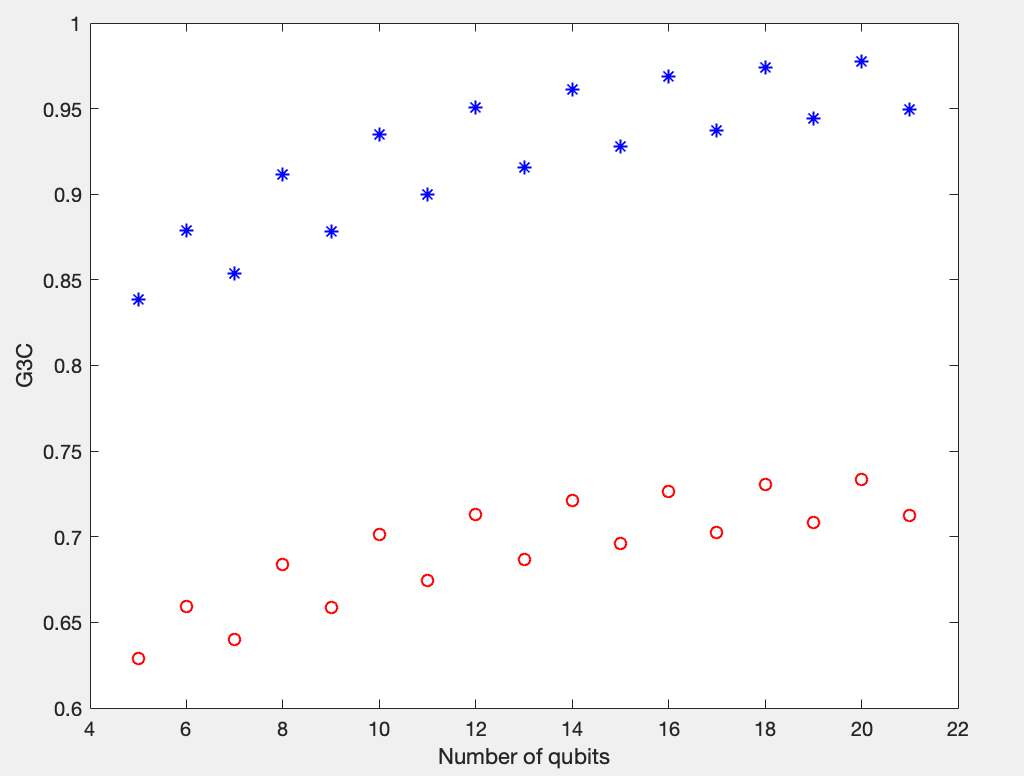}\\
		\caption{In this figure, we plot the $\mathcal{G}_3(\ket{W_n})$ with the red circle and $\frac{\mathcal{G}_3(\ket{W_n})}{\mathcal{G}_3(\ket{GHZ_n})}$ with the star when $n$ from $5$ to 21.}\label{fig3}
	\end{figure}
\end{example}
\par
Next we show the GqC  satisfies the subadditivity and continuity for pure states. First we prove the
subadditivity.\par 
Assume $\ket{\psi_1}$ and $\ket{\psi_2}$ are pure states in a bipartite system $\mathcal{H}_{d_1}\otimes\mathcal{H}_{d_2},$ then we have 
\begin{align*}
&\mathcal{C}_q(\ket{\psi_1}\otimes\ket{\psi_2})-\mathcal{C}_q(\ket{\psi_1})-\mathcal{C}_q(\ket{\psi_2})\nonumber\\
&= 1-\tr\rho_A^q\tr\sigma_A^q-1+\tr\rho_A^q-1+\tr\sigma_A^q\nonumber\\
&= -(1-\tr\rho_A^q)(1-\tr\sigma_A^q)\le0,
\end{align*}
here we denote $\rho_A=tr_B\ket{\psi_1}_{AB}\bra{\psi_1},$ $\sigma_A=tr_B\ket{\psi_2}_{AB}\bra{\psi_2}$. Due to the above inequality, we have $\mathcal{C}_q(\cdot)$ is subadditivity for pure states in arbitrary dimensional bipartite systems.
\begin{align}
\mathcal{C}_q(\ket{\psi_1}\otimes\ket{\psi_2})\le \mathcal{C}_q(\ket{\psi_1})+\mathcal{C}_q(\ket{\psi_2}).\label{sub}
\end{align}
Then we show that $\mathcal{G}_q(\cdot)$ satisfies the subadditivity property for pure states.
\par Assume $\ket{\psi_1}$ and $\ket{\psi_2}$ are two pure states in $n$-partite systems, then we have
\begin{widetext}
\begin{align}
&(\mathcal{G}_q(\ket{\psi_1}\otimes\ket{\psi_2})^{c(\alpha)}-(\mathcal{G}_q(\ket{\psi_1})+\mathcal{G}_q(\ket{\psi_2}))^{c(\alpha)}\nonumber\\
=&\Pi_{\alpha_i\in\alpha}\mathcal{C}_{qA_{\alpha_i}B_{\alpha_i}}(\ket{\psi_1}\otimes\ket{\psi_2})-[(\Pi_{\alpha_i\in\alpha}\mathcal{C}_{qA_{\alpha_i}B_{\alpha_i}}(\ket{\psi_1}))^{\frac{1}{c(\alpha)}}+(\Pi_{\alpha_i\in\alpha}\mathcal{C}_{qA_{\alpha_i}B_{\alpha_i}}(\ket{\psi_2}))^{\frac{1}{c(\alpha)}}]^{c(\alpha)}\nonumber\\
\le&\Pi_{\alpha_i\in\alpha}(\mathcal{C}_{qA_{\alpha_i}B_{\alpha_i}}(\ket{\psi_1})+\mathcal{C}_{qA_{\alpha_i}B_{\alpha_i}}(\ket{\psi_2}))-[(\Pi_{\alpha_i\in\alpha}\mathcal{C}_{qA_{\alpha_i}B_{\alpha_i}}(\ket{\psi_1}))^{\frac{1}{c(\alpha)}}+(\Pi_{\alpha_i\in\alpha}\mathcal{C}_{qA_{\alpha_i}B_{\alpha_i}}(\ket{\psi_2}))^{\frac{1}{c(\alpha)}}]^{c(\alpha)}\nonumber\\
\le&0, \label{subp}
\end{align}
\end{widetext}
here the first inequality is due to the subadditivity of the bipartite entanglement measure $(\ref{sub}),$ the second inequality is due to the Mahler's inequality.  Due to the inequality (\ref{subp}), we have $\mathcal{G}_q(\cdot)$ satisfies the subadditivity for pure states.
\par
At last, we present that the GqC satisfies continuity for pure states. First we present a lemma on the $\mathcal{C}_q(\cdot)$ of pure states.
\begin{Lemma}\label{l0}
	Assume $\ket{\psi_1}$ and $\ket{\psi_2}$ are pure states in $\mathcal{H}_{d}\otimes\mathcal{H}_{d},$ when $\norm{\ket{\psi_1}-\ket{\psi_2}}_1\le \epsilon$, then we have
	\begin{align}
	|\mathcal{C}_q(\ket{\psi_1})-\mathcal{C}_q(\ket{\psi_2})|\le d[(1+\frac{\epsilon}{d})^q-1].
	\end{align}
\end{Lemma} \par
 Then we can generalize the results to the GqC for the pure states.
\begin{Theorem}\label{th6}
	Assume $\ket{\psi_1}$ and $\ket{\psi_2}$ are two pure states in $n$-partite systems $\mathcal{H}_{d}\otimes\mathcal{H}_{d}\otimes\cdots \otimes\mathcal{H}_{d}$, here $\norm{\ket{\psi_1}-\ket{\psi_2}}\le \epsilon,$ then we have
	\begin{align}
	|\mathcal{G}_q(\ket{\psi_1})-\mathcal{G}_q(\ket{\psi_2})|\le [\sum_{i=1}^{\frac{n-1}{2}}C_n^id^i[(1+\frac{\epsilon}{d^i})^q-1]]^{\frac{1}{c(\alpha)}}
	\end{align}
\end{Theorem}
Here we place the proof of the lemma \ref{l0} and Theorem \ref{th6} in Sec. \ref{ap4}.

\subsection{A lower bound of GqC}\label{s2}
\indent In this subsection, we first present a lower bound of the entanglement measure $\mathcal{C}_q(\rho_{AB})$ for a bipartite mixed state $\rho_{AB}$, then we extend the results to the GqC for multipartite mixed states.  
\begin{Theorem}\label{th4}
For a bipartite mixed state $\rho$ on the system $\mathcal{H}_{m}\otimes\mathcal{H}_{n} (m\le n)$, $\ket{\phi}$ is an arbitrary pure state in $\mathcal{H}_{m}\otimes\mathcal{H}_{n}$ its revised parametrized entanglement measure $\mathcal{C}_q(\rho)$ satisfies 
\begin{align}
\mathcal{C}_q(\rho)\ge {co[R(\Lambda)]},
\end{align}
where $\Lambda=\max\{\frac{\bra{\phi}\rho\ket{\phi}}{s_1m},\frac{1}{m}\}$,  $R(\Lambda)=1-\gamma(\Lambda)^q-\frac{[1-\gamma(\Lambda)]^q}{(m-1)^{q-1}},$ with $\gamma(\Lambda)=\frac{\sqrt{\Lambda}+\sqrt{(m-1)(1-\Lambda)}}{m}$, and $$co[R(\Lambda)]=\frac{m^{q-1}-1}{m^{q-2}(m-1)}(\Lambda-\frac{1}{m}).$$
\end{Theorem}\par
We place the proof of this theorem in the Appendix \ref{ap2}. \par 
Before generalizing the bound for $\mathcal{C}_q(\rho_{AB})$ of a bipartite state $\rho_{AB}$ to $\mathcal{G}_q(\rho_{A_1A_2\cdots A_n})$ of a multipartite mixed state $\rho_{A_1A_2\cdots A_n}$, we will denote some definitions. Let $\ket{\psi}_{A_1A_2\cdots A_n}$ be an arbitrary $n$-partite pure state in $\mathcal{H}_{A_1}\otimes\mathcal{H}_{A_2}\otimes\cdots\otimes \mathcal{H}_{A_n},$ $\alpha_i$ be a possible bipartition of $\{A_1,A_2,\cdots,A_n\},$ then by the Schmidt decomposition, $\ket{\psi}_{A_1A_2\cdots A_n}$ can be written as $\ket{\psi}=U_{\alpha}\otimes U_{\overline{\alpha}}\sum_{i=1}^{q_{\alpha}}\sqrt{s_{i}^{(\alpha)}}\ket{ii}$ under the bipartition $\alpha|\overline{\alpha},$ here $\{\sqrt{s_{i}^{(\alpha)}}\}$ are the Schmidt coefficients in decreasing order, $m_{\alpha}$ denotes the number of nonzero Schmidt coefficients. Next let $s_1=\max_{\alpha}\{s_1^{\alpha}\},$ $m=\max_{\alpha}\{m_{\alpha}\}.$
\begin{Theorem}\label{th5}
	Assume $\rho_{A_1A_2\cdots A_n}$ is a mixed state on an $n$-paritite system. Then we have
	\begin{align}
\mathcal{G}_q(\rho)\ge \frac{m^{q-1}-1}{m^{q-2}(m-1)}(\Lambda^{'}-\frac{1}{m}),
\end{align}
here we denote $\Lambda^{'}=\max\{\frac{\bra{\phi}\rho\ket{\phi}}{s_1m},\frac{1}{m}\}$
\end{Theorem}
We place the proof of Theorem $\ref{th5}$ in the Sec. \ref{ap3}. Next we present an example to show the results.
\begin{example}
	Consider a $3$-qubit $W$ state with the white noise,
	\begin{align}
	\rho_W=p\ket{W}\bra{W}+\frac{1-p}{8}\mathbb{I},\hspace{3mm} p\in(0,1).
	\end{align}
	here we present the bound of $\mathcal{G}_2(\cdot)$ for $\rho_W$ in Fig. \ref{fig1}. There we plot the lower bound of G$2$C for $\rho_W$.
	\begin{figure}
		\centering
		\includegraphics[width=90mm]{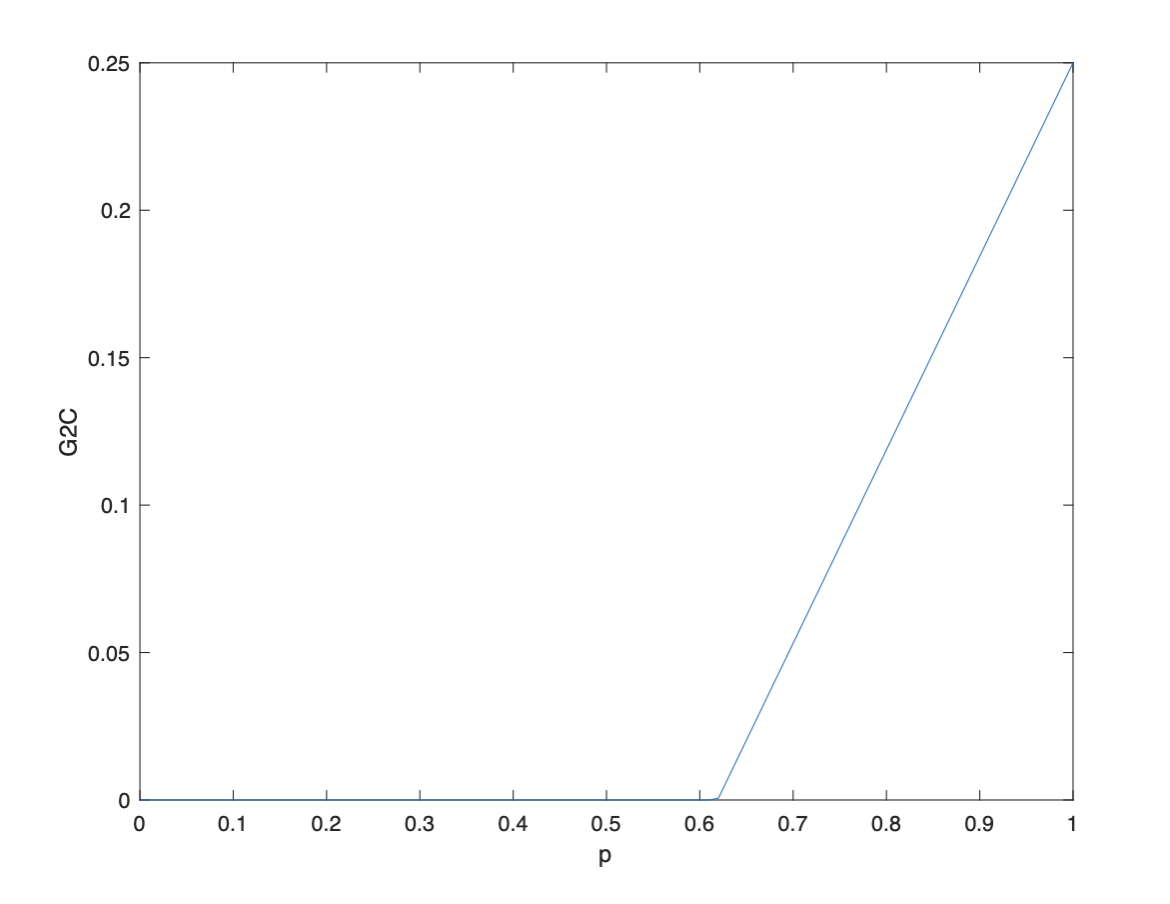}\\
		\caption{In this figure,  the red line denotes the lower bound of G2C for $\rho_W$, here we use $\ket{\phi}=\ket{W}$ and $s_1=\frac{2}{3}.$ }\label{fig1}
	\end{figure}
\end{example}
\begin{example}
	Let $\rho_{GHZ}$ be a $3$-qubit mixed state,
	\begin{align}
	\rho_{GHZ}=&c\ket{GHZ}\bra{GHZ}+\frac{1-c}{8}\mathbb{I},\nonumber\\
	\ket{GHZ}=&\frac{1}{\sqrt{2}}(\ket{000}+\ket{111}),\hspace{3mm} c\in(0,1).
	\end{align}
		here we present the bound of $\mathcal{G}_q(\cdot)$ for $\rho_W$ in Fig. \ref{fig1}. There we plot the lower bound of G$q$C for $\rho_{GHZ}$ when $q\in[2,12]$. From the Fig. \ref{fig2}, we see that the lower bound we present is monotone in terms of $q$.
		\begin{figure}[H]
			\centering
			\includegraphics[width=90mm]{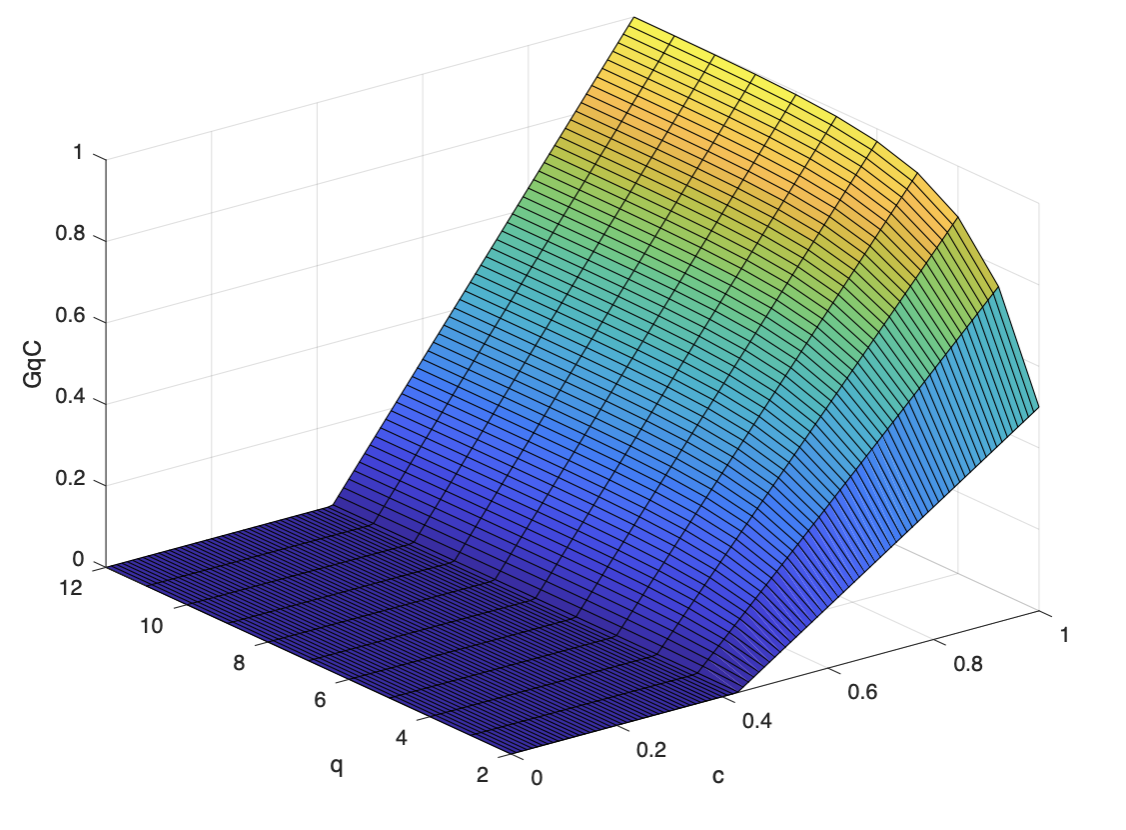}\\
			\caption{In this figure, we plot the lower bound of $G_q(\rho_{GHZ})$ when $q\in[2,12]$, here we use $\ket{\phi}=\ket{GHZ}.$}\label{fig2}
		\end{figure}
\end{example}
\subsection{Comparisons with other GME measures}\label{s3}
\indent In this section, we present some examples on the $\mathcal{G}_q(\cdot)$ and compare them with other GME measures, GBC and GGM. Through comparison, we can get the difference between GqC and other GME measures, also we can get the advantages of GqC.\\
\indent Entanglement ordering is meaningful when considering the entanglement measures \cite{virmani2001optimal,zyczkowski2002relativity}. It means that if $E_1$ and $E_2$ are two entanglement measures, for any pair of $\sigma_1$ and $\sigma_2$, $E_1(\sigma_1)\ge E_1(\sigma_2)$ derives $E_2(\sigma_1)\ge E_2(\sigma_2).$ The GqC can lead to different entanglement ordering when comparing with other GME measures. Next we consider the following two classes of states in three-qubit systems to show that the entanglement ordering of G4C is different from GMC and GGM. 
\begin{example}
	\begin{align*}
	\textit{Class I:} \hspace{2mm}&\ket{\psi}=\frac{1}{\sqrt{2}}(\cos\theta\ket{000}+\sin\theta\ket{001})+\frac{1}{\sqrt{2}}\ket{111},\\
\textit{Class II:} \hspace{2mm}&\ket{\phi}=\cos\theta\ket{000}+\sin\theta\ket{111}.	
	\end{align*}
		\begin{figure}[H]
		\centering
				\subfigure[The G4C and GMC for the Class I states and Class II states.]{\label{fig4}
			\includegraphics[width=90mm]{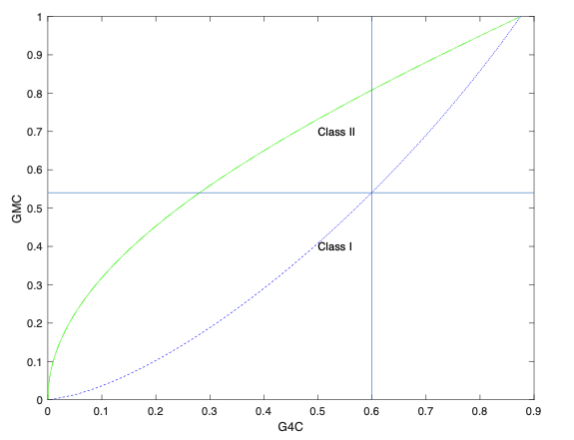}}
		\subfigure[The G4C and GGM for the Class I states and Class II states.]
{	\label{fig5}
				\includegraphics[width=90mm]{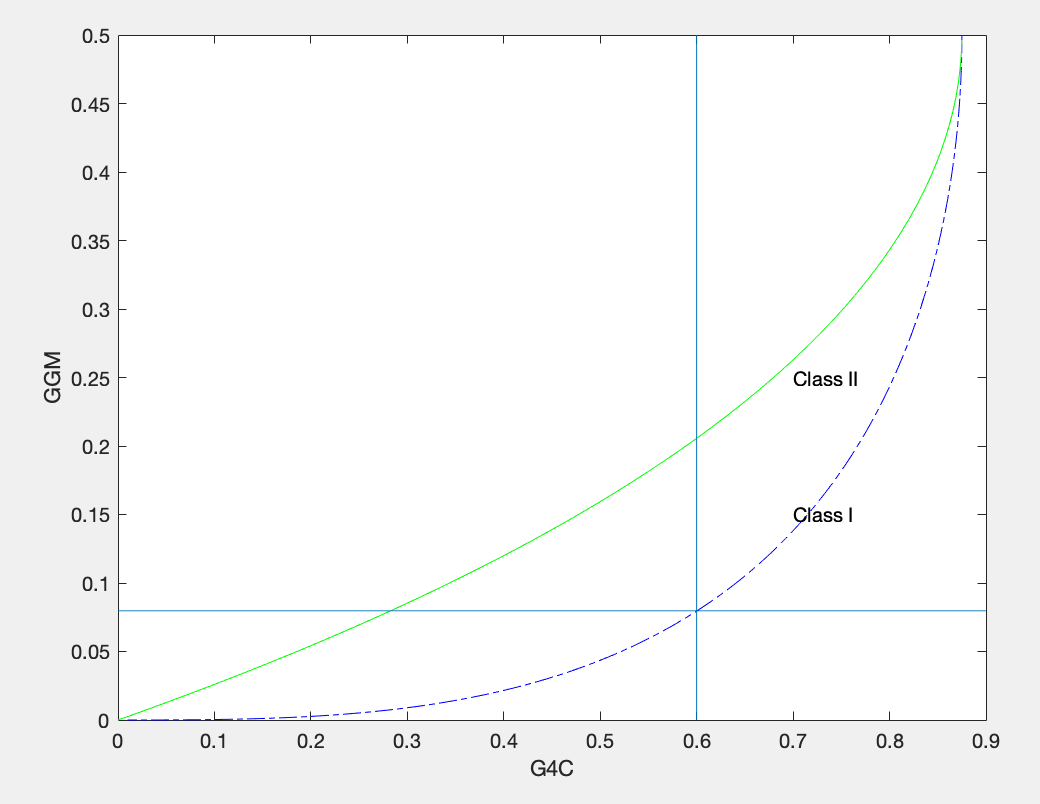}}
			\caption{The G4C, GMC and GGM for the Class I states and Class II states.}
\end{figure}
	\end{example}\par
From the Fig. \ref{fig4}, we see that the G4C owns different entanglement order from the GMC.  For a given state belonging to the Class I state, there are many states in Class II with larger GMC but with smaller G4C. This can be shown by drawing vertical or horizontal lines, then compare the states at the intersection point. Similarly, from the Fig. \ref{fig5}, the entanglement order of G4C is different from GGM. For a given state of the Class I, there are many states in Class II that with larger GGM but with smaller G4C. The opposite results can be arrived at when given a Class II pure state.\\
\indent At last, we present another class of 4-qubit pure states, which shows the advantages of GqC when comparing with GGM and GMC.

	\begin{align}
	\ket{\psi}=&\cos\theta\ket{\phi_1}+\sin\theta\ket{0111},\hspace{3mm} \theta\in(0,\frac{\pi}{2}),\label{e12}\\
	\ket{\phi_1}=&\cos\frac{2\pi}{3}\ket{0100}+\sin\frac{2\pi}{3}\ket{1000}.\nonumber
	\end{align}
		\begin{figure}[H]
		\centering
		\includegraphics[width=85mm]{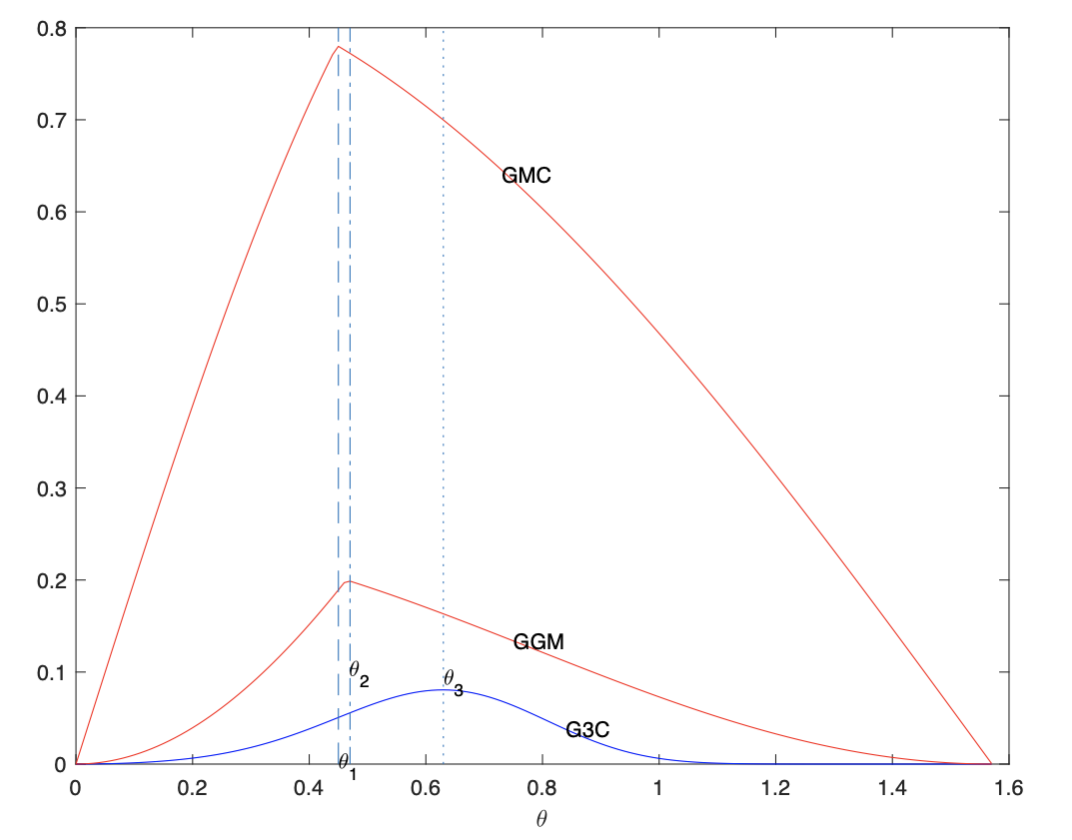}\\
		\caption{The G3C,GMC and GGM for the four qubit pure states (\ref{e12}).}\label{f7}
	\end{figure}\par
As presented in Fig.\ref{f7}, when $\theta$ increases from $[\theta_1,\theta_3],$ the values of G3C increases, while the GGM decreases from $[\theta_2,\theta_3]$ and GMC decreases from $[\theta_1,\theta_3].$ Next when $\theta$ ranges in $[0,\theta_3],$ each G3C value corresponds to a unique state in the class $(\ref{e12}),$ while there exists pairs of states in the class $(\ref{e12})$ with the same GGM or GBC, that is, the GME measure G3C can detects the robustness between some states while the GGM and GBC cannot. Moreover, Fig. \ref{f7} show the smoothness of G3C, a sharp peak appears with the varying $\theta$ when considering the GMC and GGM.
\par
\section{Conclusion}
\indent In this paper, we have proposed and investigated a GME measure based on the geometric mean method. First we have presented the GqC is a GME measure and satisfies the subadditivity and continuity for pure states. We have also presented the analytical expressions of GqC for the W states and GHZ states in $n$-qubit systems. From the analytical expressions, we can see that the entanglement of the GHZ state is stronger than the W state. Next we have presented a lower bound of the GqC based on the method proposed in \cite{dai2020experimentally}. At last, we have presented the advantages of the GqC by comparing with the GMC and GGM through some examples. Due to the importance of the study of GME measures, our results can provide a reference for future work on the study of multiparty quantum entanglement.

\bibliographystyle{IEEEtran}
\bibliography{ref}
\section{Appendix}

\subsection{Some results on $\mathcal{G}_q(\ket{W_n})$ and $\mathcal{G}_q(\ket{GHZ_n})$}\label{ap1}
	An $n$-qubit W state can be represented as
	\begin{align*}
	\ket{W_n}=\frac{1}{\sqrt{n}}(\ket{00\cdots 1}+\ket{01\cdots 0}+\cdots+\ket{00\cdots1}).
	\end{align*}
	Through computation, we have
\begin{align*}
\mathcal{C}_{q1,n-1}(\ket{W_n})=&1-\frac{1}{n^q}-\frac{(n-1)^q}{n^q},\nonumber\\
\mathcal{C}_{q,2,n-2}(\ket{W_n})=&1-\frac{2^q}{n^q}-\frac{(n-2)^q}{n^q},\nonumber\\
\cdots,&\hspace{2mm}\cdots,\hspace{2mm}\cdots\nonumber\\
\mathcal{C}_{q,k,n-k}(\ket{W_n})=&1-\frac{k^q}{n^q}-\frac{(n-k)^q}{n^q},
\end{align*} 
	An $n$-qubit GHZ state can be represented as
\begin{align*}
\ket{GHZ_n}=\frac{1}{\sqrt{2}}(\ket{00\cdots 0}+\ket{11\cdots1}).
\end{align*}
Through computation, we have
\begin{align*}
\mathcal{C}_{q,k,n-k}(\ket{GHZ_n})=1-\frac{1}{2^{q-1}},
\end{align*}
then we have
\begin{widetext}
	\begin{align*}
	\mathcal{G}_q(\ket{W_n})=
	\begin{split}
	\left\{
	\begin{array}{lr}
	(\Pi_{k=1}^{\frac{n-1}{2}}[1-\frac{k^q}{n^q}-\frac{(n-k)^q}{n^q}]^{C_n^k})^{\frac{1}{c(\alpha)}},\hspace{3mm} \textit{if $n$ is odd,}\\
	(\Pi_{k=1}^{\frac{n}{2}-1}[1-\frac{k^q}{n^q}-\frac{(n-k)^q}{n^q}]^{C_n^k}[1-\frac{1}{2^{q-1}}]^{\frac{C_n^{\frac{n}{2}}}{2}})^{\frac{1}{c(\alpha)}},\hspace{3mm}\textit{if $n$ is even.}
	\end{array}
	\right.
	\end{split}
	\end{align*}
	\begin{align*}
	\frac{\mathcal{G}_q(\ket{W_n})}{\mathcal{G}_q(\ket{GHZ_n})}
	=\begin{split}
	\left\{
	\begin{array}{lr}
	\frac{\exp[\frac{\sum_{k=1}^{\frac{n-1}{2}}C_n^k}{c(\alpha)}(\ln(1-\frac{k^q}{n^q}-\frac{(n-k)^q}{n^q}))]}{1-\frac{1}{2^{q-1}}},\hspace{3mm} \textit{if $n$ is odd,}\\
	\frac{\exp[\frac{\sum_{k=1}^{\frac{n}{2}-1}C_n^k}{c(\alpha)}\ln(1-\frac{k^q}{n^q}-\frac{(n-k)^q}{n^q})+\frac{\frac{C_n^{\frac{n}{2}}}{2}\ln(1-\frac{1}{2^{q-1}})}{c(\alpha)}]}{1-\frac{1}{2^{q-1}}},\hspace{3mm}\textit{if $n$ is even.}
	\end{array}
	\right.
	\end{split}
	\end{align*}
\end{widetext}
Then by using the Stolz-Cresaro theorem, we have 
$$\lim_{k\rightarrow\infty}\frac{\mathcal{G}_q(\ket{W_{2k}})}{\mathcal{G}_q(\ket{GHZ_{2k}})}=1,$$
$$\lim_{k\rightarrow\infty}\frac{\mathcal{G}_q(\ket{W_{2k+1}})}{\mathcal{G}_q(\ket{GHZ_{2k+1}})}=1.$$
\subsection{The proof on the continuity of $\mathcal{G}_q(\cdot)$ for pure states} \label{ap4}
Here we present the proof of Lemma \ref{l0}.\par
Lemma \ref{l0}:	Assume $\ket{\psi_1}$ and $\ket{\psi_2}$ are pure states in $\mathcal{H}_{d}\otimes\mathcal{H}_{d},$ when $\norm{\ket{\psi_1}-\ket{\psi_2}}_1\le \epsilon$, then we have
\begin{align}
|\mathcal{C}_q(\ket{\psi_1})-\mathcal{C}_q(\ket{\psi_2})|\le d[(1+\frac{\epsilon}{d})^q-1].
\end{align}
\begin{proof}
	As partial trace is trace-preserving, then $\norm{\rho_A-\sigma_A}_1\le \epsilon,$ here $\rho_A=Tr_B\ket{\psi}\bra{\psi}$, $\sigma_A=Tr_B\ket{\sigma}\bra{\sigma}.$\par
	Next as $\norm{\cdot}_1$ is unitarily invariant norm, then
	\begin{align*}
	\norm{Eig^{\downarrow}(\rho_A)-Eig^{\downarrow}(\sigma_A)}\le \norm{\rho_A-\sigma_A}_1,\nonumber\\
	\norm{\rho_A-\sigma_A}_1 \le \norm{Eig^{\downarrow}(\rho_A)-Eig^{\uparrow}(\sigma_A)},
	\end{align*}
	that is, we only need to consider the classical case. Readers who are interesting in the above two inequalities please refer to \cite{bhatia2013matrix}.\par
	Assume $\rho_A$ and $\sigma_A$ are two diagonal density matrices with their diagonal elements are $\{p_i\}_{i=1}^d$ and $\{r_i\}_{i=1}^d,$ respectively. And $\{p_i\}_{i=1}^d$ and $\{r_i\}_{i=1}^d$ satisfy $\sum_ip_i=\sum_ir_i=1,$ and $\sum_i|p_i-r_i|\le\epsilon.$ Then 
	\begin{align}
	|\mathcal{C}_q(\rho_A)-\mathcal{C}_q(\sigma_A)|=&|\tr\rho_A^q-\tr\sigma_A^q|\nonumber\\
	\le &\sum_i|p_i^q-r^q_i|\label{in1}
	\end{align}
	Next let $|p_i-r_i|=\epsilon_i$, $\sum_i\epsilon_i\le \epsilon,$ when $p_i=r_i+\epsilon_i,$ then 
	\begin{align*}
	|p_i^q-r_i^q|=&|(r_i+\epsilon_i)^q-r_i^q|\nonumber\\
	\le& (1+\epsilon_i)^q-1,
	\end{align*}
	the last inequality is due to that $(r_i+\epsilon_i)^q-r_i^q$ is increasing in terms of $r_i$, and $r_i\in(0,1).$ And due to that  $(r_i+\epsilon_i)^q-r_i^q$ is increasing in terms of $r_i$, when $p_i=r_i-\epsilon_i$, the above inequality is also valid. Then the inequality $(\ref{in1})$ becomes 
	\begin{align}
	(\ref{in1})=\sum_i|p_i^q-r^q_i|\le& \epsilon^q\frac{\sum_i[(1+\epsilon_i)^q-1]}{\epsilon^q}\nonumber\\
	\le& d[(1+\frac{\epsilon}{d})^q-1],
	\end{align}
	when all $\epsilon_i=\frac{\epsilon}{d}$, the equality in the last inequality is valid.
\end{proof}\par
Theorem \ref{th6}:	Assume $\ket{\psi_1}$ and $\ket{\psi_2}$ are two pure states in $n$-partite systems $\mathcal{H}_{d}\otimes\mathcal{H}_{d}\otimes\cdots \otimes\mathcal{H}_{d}$, here $\norm{\ket{\psi_1}-\ket{\psi_2}}\le \epsilon,$ then we have
\begin{align}
|\mathcal{G}_q(\ket{\psi_1})-\mathcal{G}_q(\ket{\psi_2})|\le [\sum_{i=1}^{\frac{n-1}{2}}C_n^id^i[(1+\frac{\epsilon}{d^i})^q-1]]^{\frac{1}{c(\alpha)}}
\end{align}
\begin{proof}
	When $n$ is odd, then we have 
	\begin{align}
	&|\mathcal{G}_q(\ket{\psi_1})-\mathcal{G}_q(\ket{\psi_2})|\nonumber\\
	=&|\mathcal{P}_q(\ket{\psi_1})^{\frac{1}{c(\alpha)}}-\mathcal{P}_q(\ket{\psi_2})^{\frac{1}{c(\alpha)}}|\nonumber\\
	\le&|\mathcal{P}_q(\ket{\psi_1})-\mathcal{P}_q(\ket{\psi_2})|^{\frac{1}{c(\alpha)}}\nonumber\\
	=&|\Pi_{\alpha_i\in\alpha}\mathcal{C}_{qA_{\alpha_i}B_{\alpha_i}}(\ket{\psi_1})-\Pi_{\alpha_i\in\alpha}\mathcal{C}_{qA_{\alpha_i}B_{\alpha_i}}(\ket{\psi_2})|^{\frac{1}{c(\alpha)}}\nonumber\\
	\le&[\sum_{i=1}^{\frac{n-1}{2}}C_n^id^i[(1+\frac{\epsilon}{d^i})^q-1]]^{\frac{1}{c(\alpha)}}.
	\end{align}
	the first inequality is due to that when $p,q,x\in(0,1)$, $|p^x-q^{x}|\le |p-q|^x$, the second inequality is due to the following inequality, when $x_i,y_i\in(0,1),$ $i=1,2,\cdots,n,$ then we have
	\begin{align}
	&|x_1x_2\cdots x_n-y_1y_2\cdots y_n|\nonumber\\
	=& |(x_1-y_1)x_2x_3\cdots x_n+y_1(x_2-y_2)x_3\cdots x_n\nonumber\\
	\hspace{1mm}+&y_1y_2(x_3-y_3)x_4x_5\cdots x_n\cdots+y_1y_2\cdots y_{n-1}(x_n-y_n)|\nonumber\\
	\le& \sum_i|x_i-y_i|,
	\end{align}
	here the last inequality is due to the triangle inequality and $x_i,y_i\in(0,1).$
\end{proof}
\subsection{The proof of Theorem \ref{th4}}\label{ap2}
Here we prove Theorem \ref{th4} based on the method in \cite{zhang2016evaluation}.\par
Theorem \ref{th4}:
\emph{	For a bipartite mixed state $\rho$ on the system $\mathcal{H}_{m}\otimes\mathcal{H}_{n} (m\le n)$, its revised parametrized entanglement measure $\mathcal{C}_q(\rho)$ satisfies 
	\begin{align}
	\mathcal{C}_q(\rho)\ge co[R(\Lambda)],
	\end{align}
	where $\Lambda=\max\{\frac{\bra{\phi}\rho\ket{\phi}}{s_1m},\frac{1}{m}\}$, $R(\Lambda)=1-\gamma(\Lambda)^q-\frac{[1-\gamma(\Lambda)]^q}{(m-1)^{q-1}},$ with $\gamma(\Lambda)=\frac{\sqrt{\Lambda}+\sqrt{(m-1)(1-\Lambda)}}{m}$, and $$co[R(\Lambda)]=\frac{m^{q-1}-1}{m^{q-2}(m-1)}(\Lambda-\frac{1}{m}).$$}
\begin{proof}
Here we consider the following function
\begin{align}
R(\lambda)=\min_{\vec{\mu}}\{L(\vec{\mu})|\lambda=\frac{1}{m}(\sum_i\sqrt{\mu_i})^2\}.
\end{align}
Here $L(\vec{\mu})=1-\sum_i \mu_i^2$.
Due to the results in \cite{terhal2000entanglement,berry2003bounds}, the minimum $L(\vec{\mu})$ versus $\lambda$ to $\vec{\mu}$ in the form 
\begin{align}
\vec{\mu}=\{t,\frac{1-t}{m-1},\frac{1-t}{m-1},\cdots,\frac{1-t}{m-1}\}\hspace{3mm} \textit{for $t\in[\frac{1}{m},1],$}
\end{align}
Therefore, we have the minimum $L(\vec{\mu})$ and the function $t(\lambda)$ are 
\begin{align}
L(t)=&1-t^q-\frac{(1-t)^q}{(m-1)^{q-1}},\label{t0}\\
t(\lambda)=&\frac{1}{m}(\sqrt{\lambda}+\sqrt{(m-1)(1-\lambda)})^2, \hspace{3mm} \textit{with $\lambda\in[\frac{1}{m},1].$}\label{t1}
\end{align}
Substituting $(\ref{t1})$ into (\ref{t0}), we have
\begin{align}
L^{'}_{\lambda}=L^{'}_t\times t^{'}_{\lambda},
\end{align}
through computation, we have
\begin{align}
&L^{'}_t=-qt^{q-1}+\frac{q(1-t)^{q-1}}{(m-1)^{q-1}},    \\
&L_t^{''}=-q(q-1)t^{q-2}-\frac{q(q-1)(1-t)^{q-2}}{(m-1)^{q-1}}\le 0,\\
&t_{\lambda}^{'}=\frac{1}{m}(\sqrt{\lambda}+\sqrt{(m-1)(1-\lambda)})\nonumber\\
&\hspace{20mm}\times(\frac{1}{\sqrt{\lambda}}+\frac{1-m}{\sqrt{(m-1)(1-\lambda)}})\le0,
\end{align}
as $L_t^{''}\le 0,$ $\max_t L^{'}_t=L^{'}(\frac{1}{m})=0,$ that is $L_{\lambda}^{'}\ge 0,$ $L(\lambda)$ is an increasing function. \par
Next we prove $L(\lambda)$ is concave. 
\begin{align}
L_{\lambda}^{''}=L_{t}^{''}{t_{\lambda}^{'}}^2+L^{'}_t{t_{\lambda}^{''}}
\end{align}
and
\begin{align}
t_{\lambda}^{''}=-\frac{\sqrt{(m-1)(1-\lambda)}}{2m(1-\lambda)^2\lambda^{\frac{3}{2}}}.
\end{align}
Through rough computation, 
\begin{align*}
L_{\lambda}^{''}\le 0,
\end{align*}
then we have
\begin{align}
co[R(\lambda)]=\frac{m^q-1}{m^{q-1}(m-1)}(\lambda-\frac{1}{m}).
\end{align}
\par Next we have 
\begin{align}
\mathcal{C}_q(\rho)=&\sum_i p_i \mathcal{C}_q(\ket{\psi_i})=\sum_i p_i L(\vec{\mu^j})\nonumber\\
\ge& \sum_i p_i co[R(\lambda^j)]\ge co[R(\sum_jp_j\lambda^j)]\ge co[R(\Lambda)], \label{r1}
\end{align}
where $\ket{\psi_j}=U_A\otimes U_B\sum_i \sqrt{\mu_i^j}\ket{ii}$ with $\sqrt{\mu_i^j}$ being its Schmidt coefficients in decreaing order. The last inequality is due to that $co{R(\lambda)}$ is an increasing function.
\end{proof}
\subsection{The proof of Theorem \ref{th5}}\label{ap3}
Here we present the proof of Theorem \ref{th5}.\par
Theorem \ref{th5}: 	Assume $\rho_{A_1A_2\cdots A_n}$ is a mixed state on an $n$-paritite system. Then we have
\begin{align}
\mathcal{G}_q(\rho)\ge\frac{m^{q-1}-1}{m^{q-2}(m-1)}(\Lambda^{'}-\frac{1}{m}),
\end{align}
here we denote $\Lambda^{'}=\max\{\frac{\bra{\phi}\rho\ket{\phi}}{s_1m},\frac{1}{m}\}$\par
\begin{proof}
	Assume $\{p_j,\ket{\phi_j}\}$ is the optimal decompostion of the state $\rho$ in terms of $\mathcal{G}_q(\rho),$ then we have
	\begin{align}
	\mathcal{G}_q(\rho)=&\sum_j p_j[ \Pi_{\alpha_i\in \alpha}\mathcal{C}_{q A_{\alpha_i}B_{\alpha_i}}(\ket{\phi_j})]^{\frac{1}{c(\alpha)}}\nonumber\\
	=&\sum_j [ \Pi_{\alpha_i\in \alpha}p_j\mathcal{C}_{q A_{\alpha_i}B_{\alpha_i}}(\ket{\phi_j})]^{\frac{1}{c(\alpha)}}\nonumber\\
	\ge&\sum_j \min_{\alpha_i\in\alpha}p_j\mathcal{C}_{qA_{\alpha_i}B_{\alpha_i}}(\ket{\phi_j})\nonumber\\
	\ge&\min_{\alpha_i\in\alpha}\sum_j p_j\mathcal{C}_{qA_{\alpha_i}B_{\alpha_i}}(\ket{\phi_j})\nonumber\\
	\ge& \frac{m^{q-1}-1}{m^{q-2}(m-1)}(\Lambda^{'}-\frac{1}{m}),
	\end{align}
	here we denote $\Lambda^{'}=\max\{\frac{\bra{\phi}\rho\ket{\phi}}{s_1m},\frac{1}{m}\}.$ The first inequality is due to that $\min(x,y)\le \sqrt{xy},$ $\forall x,y \in[0,1].$ The last inequality is due to the Lemma \ref{l0}.
\end{proof}
\end{document}